\documentclass[conference,letterpaper]{IEEEtran}
\usepackage[utf8]{inputenc} 
\usepackage[T1]{fontenc}
\usepackage{url}
\usepackage{color}
\usepackage{amsthm}
\usepackage{ifthen}
\usepackage{cite}
\usepackage{multirow}
\usepackage{array}
\usepackage{booktabs}
\usepackage{graphicx}
\usepackage{subcaption}
\usepackage{ifthen}

\usepackage[cmex10]{amsmath} 
\DeclareMathAlphabet\mathbfcal{OMS}{cmsy}{b}{n}
\DeclareFontFamily{OT1}{pzc}{}
\DeclareFontShape{OT1}{pzc}{m}{it}{<-> s * [1.10] pzcmi7t}{}
\DeclareMathAlphabet{\mathpzc}{OT1}{pzc}{m}{it}
\DeclareFontFamily{OT1}{bfpzc}{}
\DeclareFontShape{OT1}{bfpzc}{m}{it}{<-> s * [1.10] pzcmi7t}{}
\DeclareMathAlphabet{\mathbfpzc}{OT1}{bfpzc}{b}{it}

\newcommand{\bfc}
{\mathbfcal{C}}
\newcommand{\nc}
{\mathpzc{c}}
\newcommand{\bfp}
{\mathbf{r}}
\newcommand{\np}
{\mathpzc{r}}
 \newtheorem{Theorem}{Theorem}

\newtheorem{Lemma}{Lemma}

\newboolean{withappendix}
\setboolean{withappendix}{true}

\begin{document}
\title{Age of Information Optimization with Preemption Strategies for Correlated Systems} 

% %%% Single author, or several authors with same affiliation:
% \author{%
%  \IEEEauthorblockN{Author 1 and Author 2}
% \IEEEauthorblockA{Department of Statistics and Data Science\\
%                    University 1\\
 %                   City 1\\
  %                  Email: author1@university1.edu}% }

%%% Several authors with up to three affiliations:

\author{\IEEEauthorblockN{Egemen Erbayat$^1$, Ali Maatouk$^2$, Peng Zou$^3$, Suresh Subramaniam$^1$}
\IEEEauthorblockA{
\textit{$^{1}$The George Washington University}, \textit{$^{2}$Yale University}, \textit{$^{3}$Nanjing University of Information Science and Technology}\\
$^1$\textit{\{erbayat, suresh\}@gwu.edu}, $^2$\textit{ali.maatouk@yale.edu}, $^3$\textit{003967@nuist.edu.cn}} 
}

\maketitle

\begin{abstract}
In this paper, we examine a multi-sensor system where each sensor monitors multiple dynamic information processes and transmits updates over a shared communication channel. These updates may include correlated information across the various processes. In this type of system, we analyze the impact of preemption, where ongoing transmissions are replaced by newer updates, on minimizing the Age of Information (AoI). While preemption is optimal in some scenarios, its effectiveness in multi-sensor correlated systems remains an open question. To address this, we introduce a probabilistic preemption policy, where the source sensor preemption decision is stochastic. We derive closed-form expressions for the AoI and frame its optimization as a sum of linear ratios problem, a well-known NP-hard problem. To navigate this complexity, we establish an upper bound on the iterations using a branch-and-bound algorithm by leveraging a reformulation of the problem. This analysis reveals linear scalability with the number of processes and a logarithmic dependency on the reciprocal of the error that shows the optimal solution can be efficiently found. Building on these findings, we show how different correlation matrices can lead to distinct optimal preemption strategies. Interestingly, we demonstrate that the diversity of processes within the sensors' packets, as captured by the correlation matrix, plays a more significant role in preemption priority than the number of updates.

\end{abstract}

\section{Introduction}

In sensor networks, maintaining data freshness is crucial to support diverse applications such as environmental monitoring, industrial automation, and smart cities \cite{kandris2020applications}. A critical metric for quantifying data freshness is the Age of Information (AoI), which measures the time elapsed since the last received update was generated \cite{yates2012}. Minimizing AoI is essential in dynamic environments, where obsolete information can result in inaccurate decisions or missed opportunities. Efficient AoI management involves balancing update frequency, data relevance, and network resource constraints to ensure decision-makers have timely and accurate information when required \cite{yates2021age}. The significance of AoI has led to extensive research on its optimization across various domains, including single-server systems with one or multiple sources \cite{modiano2015,mm1,sun2016,najm2018,soysal2019,9137714,yates2019,zou2023costly}, scheduling strategies \cite{modiano-sch-1,9007478,sch-igor-1,9241401,sch-li,sch-sun}, and analysis of resource-constrained systems \cite{const-ulukus,const-biyikoglu,const-arafa,const-farazi,const-parisa}. 

%\ali{A good transition here would be: one particular area that has been garnering focus by the AoI researchers and that is correalted systems. In fact, sensor networks often handle...}

Among the strategies for AoI minimization, packet preemption is regarded as a cornerstone approach for ensuring the timeliness of information in communication networks, especially when resources such as service rates are limited \cite{yates2021age}. By prioritizing critical updates, preemption ensures that the most valuable data reaches its destination promptly, as demonstrated in the context of single-sensor, memoryless systems \cite{kaul2012status,inoue2019general}. Beyond this specific scenario, numerous studies have extensively investigated its role in optimizing AoI across diverse settings. For example, \cite{maatouk2019age} analyzes systems with prioritized information streams sharing a common server, where lower-priority packets may be buffered or discarded. Similarly, \cite{wang2019preempt} and \cite{kavitha2021controlling} examine preemption strategies for rate-limited links and lossy systems, identifying in the process the optimal policies for minimizing the AoI.

On the other hand, one particular area that has been garnering focus among AoI researchers is correlated systems. In fact, sensor networks often handle correlated data streams, where relationships between data collected by different sensors can be leveraged to enhance decision-making, reduce redundancy, and improve overall system performance \cite{mahmood2015reliability,yetgin2017survey}. This correlation often arises when multiple sensors monitor overlapping areas or related phenomena, allowing them to collaboratively exchange information and optimize resource usage. The role of correlation in sensor networks has further been explored in studies focusing on its potential to optimize system efficiency and effectiveness \cite{he2018,tong2022,popovski2019,modiano2022,ramakanth2023monitoring,erbayat2024}.

% The importance of AoI and correlation in sensor networks has motivated extensive research into optimizing AoI within correlated sensor systems. For example, \cite{he2018} studied sensor networks with overlapping fields of view, proposing a joint optimization framework for fog node assignment and transmission scheduling to reduce the AoI of multi-view image data. Similarly, \cite{tong2022} focused on overlapping camera networks, introducing scheduling algorithms for multi-channel systems designed to minimize AoI. Other works, such as \cite{popovski2019, modiano2022}, leveraged probabilistic correlation models to formulate sensor scheduling strategies aimed at lowering AoI. Additionally, \cite{ramakanth2023monitoring} treated the correlation of status updates as a discrete-time Wiener process, developing a scheduling policy that balances AoI reduction with monitoring accuracy. Furthermore, \cite{erbayat2024} analyzed the impact of optimal correlation probabilities under varying environmental conditions, addressing the interplay between error minimization and AoI.

%\ali{On the other hand, Preemption in AoI systems has been widely studied...Also, Id say reduce the size of this paragraph} 

%\ali{I don't like this transition here. Talk about correlated systems in the previous paragraph and how AoI is of interest. Then, switch here to preemption is still open question. Do not focus on your paper as you did here}
As part of ongoing efforts in this area, the potential of leveraging interdependencies between sensors to reduce the AoI in correlated systems has been studied, but the benefits and challenges of employing preemption in multi-sensor systems with correlated data streams remain an open question. While preemption is a potential strategy to minimize AoI in a network, it is not always the optimal strategy \cite{yates2019}. This approach must account for the specific features of the packets being transmitted since preempting leads to prioritization. For example, a sensor with a lower arrival rate may track a unique process that no other sensor monitors, making its packets particularly valuable and critical to retain. On the other hand, preempting a packet from a sensor with a high arrival rate may not significantly reduce AoI, as the frequent updates from such sensors diminish the impact of losing a single packet.

%\ali{Here you make the connection between preemption and multi-sensor correlated systems}

%\ali{Its good to emphasize that we have correlation here so it is different than typical AoI system}.

To address this gap, this paper introduces adaptable and probabilistic preemption mechanisms that dynamically balance priorities across sensors, considering their unique correlation characteristics and resource demands. To that end, the main contributions of this paper are summarized as follows:

\begin{itemize}
    \item As a first step, we propose a system where the ability of a packet to preempt an ongoing transmission probabilistically depends on its source rather than being fixed or following deterministic rules. Subsequently, using stochastic hybrid system modeling, we derive a closed-form expression for AoI to analyze the impact of probabilistic preemption on network performance.
    
    %enabling a more adaptable approach to manage updates by giving higher-priority updates a better chance to preempt, ensuring information remains up-to-date.

    \item Following that, we investigate optimizing the total AoI in multi-sensor systems, considering the interplay between diverse sensors and shared resources. Building on this, we frame the problem of deciding optimal preemption strategies as a sum of linear ratios problem, which is generally an NP-Hard problem\cite{freund2001solving}. However, by analyzing its unique characteristics, we derive an upper bound on the number of iterations required to identify optimal preemption strategies using a branch-and-bound algorithm, thus ensuring computational efficiency in finding the optimal solution.
    %\ali{You are using a lot the , ensuring... it sounds very chatgpt liky, try to minimize those when possible. Also, talk about the bounds and the impact of these results on getting an efficient solution}
    \item Lastly, we validate our findings with numerical results and evaluate optimal preemption strategies to minimize AoI. Our findings demonstrate how correlation influences preemption strategies. Notably, when a source provides a lower aggregate number of updates while distributing them more evenly, the system prioritizes it for preemption, even if another sensor updates less frequently.\ifthenelse{\boolean{withappendix}}
{}
{\footnote{Due to space limitations, we present the proof details in \cite{technicalNote}.}}
 %\ali{Dont forget to put the right link}
\end{itemize}

%These results not only support the theory but also offer practical ideas for real-world use, such as in IoT networks, factories, or autonomous systems, where staying up-to-date is very important.

%The remainder of this paper is structured as follows. Section \ref{system-model} introduces the system model and key assumptions. In Section \ref{aoi-S}, we derive the closed-form expression for the AoI within the proposed system. Section \ref{aoi-opt} outlines the optimization problem and details the process of determining the optimal preemption probabilities. The numerical results are presented in Section \ref{numerical}, and the paper concludes with a summary and discussion in Section \ref{conc}.

\section{System Model}\label{system-model}

We consider a sensor network in which \(N\) sensors monitor \(M\) information processes. 
Each process is dynamic, with its state evolving over time. To ensure the monitor remains updated, each sensor generates status updates and transmits them through a shared server, as illustrated in Figure \ref{fig:system_model}. 
\begin{figure}[!t]
  \centering
  \includegraphics[width=0.35\textwidth]{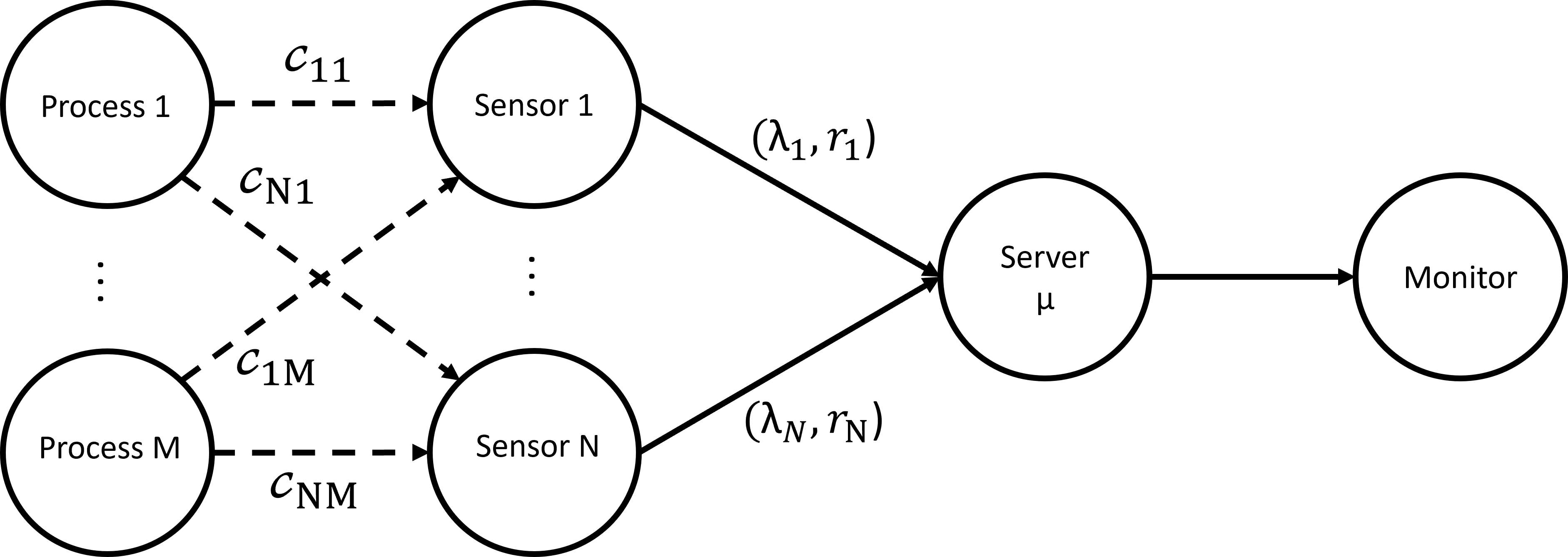}
  \caption{Illustration of our system model.}
  \vspace{-14pt}\label{fig:system_model}
\end{figure}
We assume that the service time of each packet follows an exponential distribution with a service rate \(\mu\). Additionally, sensor \(i\) generates packets following a Poisson process with an arrival rate of \(\lambda_i\). We adopt a zero-buffer with probabilistic preemption assumption, motivated by its effectiveness in minimizing AoI in systems with preemption under specific conditions \cite{bedewy}. While this may not strictly apply to our scenario, our findings, \cite{erbayat2024}, show that buffers do not consistently improve performance when there is no preemption. Furthermore, the zero-buffer assumption is intuitive with preemption, as arriving packets either preempt the one in service or are dropped. Thus, we maintain this assumption throughout our analysis, implying that an arriving packet finding the server busy is either dropped or preempts the packet in the server \cite{dataNetworks:book}, depending on the adopted preemption policy.
%\ali{I feel like we talk a lot here, let us reduce this a bit and just reference our previous work}

With these assumptions, we define \(\boldsymbol{\lambda}\) as the vector of arrival rates from different sensors, where \(\lambda_i\) represents the arrival rate from sensor \(i\) for \(i=1,\ldots,N\). Specifically, we express \(\boldsymbol{\lambda}\) as:
\begin{equation}
\boldsymbol{\lambda}^T = \begin{bmatrix}
\lambda_{1} & \lambda_{2} & \dots & \lambda_{N}
\end{bmatrix}.
\end{equation}

To model the correlation between sensor observations, we assume that a packet generated by sensor \(i\) contains information about process \(j\) with a correlation probability \(\nc_{ij}\). This information reflects the state of the process at the time of packet generation. The correlation matrix \( \bfc \in [0,1]^{N \times M}\) is defined as
\begin{equation} 
\bfc = \begin{bmatrix}
\nc_{11} & \nc_{12} & \dots & \nc_{1M} \\
\nc_{21} & \nc_{22} & \dots & \nc_{2M} \\
\vdots & \vdots & \ddots & \vdots \\
\nc_{N1} & \nc_{N2} & \dots & \nc_{NM}
\end{bmatrix}.
\end{equation}

Each sensor is also associated with a preemption probability, which defines the likelihood that a packet arriving from sensor \(i\) preempts the packet currently being served when the server is busy.\footnote{To be more general, the preemption probabilities may depend on both the sensor that generated the arriving packet and the sensor that generated the packet in service. We defer the examination of this case for future work.} To represent the preemption probabilities of all sensors, we use the vector \(\bfp \in [0,1]^{N}\), where the \(i\)-th entry, \(\np_i\), denotes the preemption probability corresponding to sensor \(i\). The vector \(\bfp\) is expressed as
\begin{equation}
\bfp^T = \begin{bmatrix}
\np_{1} & \np_{2} & \dots & \np_{N} 
\end{bmatrix}.
\end{equation}
%\ali{Nx1 is a vector not a matrix. Let use a vector instead. Also, is correlation fixed by sensor, and not for each sensor i and process m there $r_{im}$?} \suresh{Ali, each packet contains correlated info from multiple processes, so the preemption prob is associated with the sensor. Egemen, could we improve it by making the pre-prob dependent on both the new packet sensor and the sensor of the packet-in-service? We should at least acknowledge that possibility even if we don't do it in this paper. Also, the system model says almost nothing about the process and its states. I think you should elaborate on that a bit.} \egemen{Correlation is not fixed by sensor that is $\nc_{ij}$. The preemption probability $\np_i$ is fixed by the sensor. We can make it pre-prob dependent and define a $NxN$ matrix. What should be the best way to include this possibility because I do not have an analysis for it?. I have just added a sentence indicating that the process' state evolves over time because we do not consider error. Should I elaborate more? }\suresh{This is fine. I added a footnote for that case.}
In the next section, we present the closed-form derivation of the AoI under the described system model.

\section{Age of Information Analysis} \label{aoi-S}

%In this section, we consider the AoI measure introduced in \cite{yates2012} as the performance metric to evaluate the timeliness of updates at the monitor. 

For analytical convenience, we first simplify the system by focusing solely on the relevance of the updates rather than their origin. In the considered system, the source sensor of the packet containing information about any arbitrary process $j$ is irrelevant from the monitor's perspective. Instead, what matters is whether the served packet contains information about process $j$, regardless of which sensor provided the update to track AoI. To formalize this, we label a status update as informative for process $j$ if it contains information on process $j$. Otherwise, it is labeled as uninformative. We build on this concept by defining two types of informative arrivals:
\begin{itemize}
    \item Informative arrivals that can preempt ongoing service.
    \item Informative arrivals that cannot preempt ongoing service.
\end{itemize}
Using this distinction, we define the informative arrival rate vectors as follows
\begin{equation}
\boldsymbol{\Tilde{\lambda}}^T = \begin{bmatrix}
\Tilde{\lambda}_{1} & \Tilde{\lambda}_{2} & \dots & \Tilde{\lambda}_{M}
\end{bmatrix} = (\boldsymbol{\lambda}^T \odot \bfp^T) \bfc,
\end{equation}
\begin{equation}
\boldsymbol{\dot{\lambda}}^T = \begin{bmatrix}
\dot{\lambda}_{1} & \dot{\lambda}_{2} & \dots & \dot{\lambda}_{M}
\end{bmatrix} = (\boldsymbol{\lambda}^T \odot (1-\bfp^T)) \bfc,
\end{equation}
where \(\boldsymbol{\Tilde{\lambda}}\) represents the informative arrival rate vector for packets that can preempt ongoing service, and \(\boldsymbol{\dot{\lambda}}\) represents the informative arrival rate vector for packets that cannot preempt ongoing service. The total channel arrival rate is given as:
\begin{equation}
\lambda_C = \Tilde{\lambda}_C + \dot{\lambda}_C,
\end{equation}
where
\begin{equation}
\Tilde{\lambda}_C = \sum_{i=1}^{N} \lambda_i \np_{i} \text{ and }
\dot{\lambda}_C = \sum_{i=1}^{N} \lambda_i (1 - \np_{i}).
\end{equation}
represent the channel arrival rates for packets that can and cannot preempt, respectively.

With the above quantities in mind, we analyze the system by reducing it to $M$ independent systems, each consisting of two sources as depicted in Figure \ref{fig:equiv_model}. The independence of these $M$ systems arises from the Poisson nature of packet arrivals. For any single process $j$, the arrivals of both informative and uninformative packets from all other processes collectively form Poisson streams, as shown in \ifthenelse{\boolean{withappendix}}
{Appendix~\ref{reduction-P}}
{Appendix A in \cite{technicalNote}}.%\ali{Can we put this in appendix instead of citing our work?}

\begin{figure}[!t]
  \centering
  \includegraphics[width=0.35\textwidth]{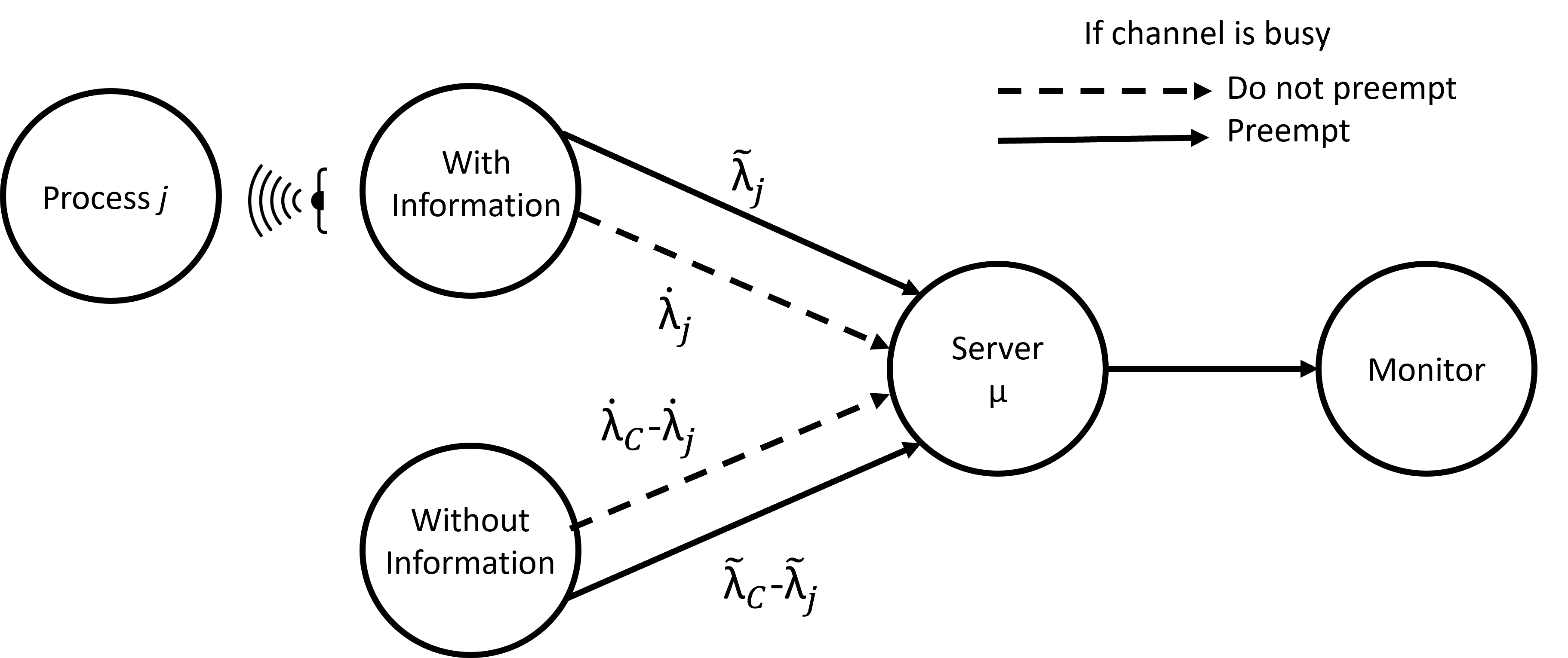}
  \caption{Equivalent system model from process $j$'s perspective.}
  \vspace{-12pt}\label{fig:equiv_model}
\end{figure}

With this setup established, we now proceed to evaluate the evolution of AoI for a single process. The AoI of process \(j\) at time \(t\), denoted by \(\Delta_j(t)\), is defined as
\begin{equation}
\Delta_j(t) = t - T_j, \quad j=1,\ldots,M,
\end{equation}
where \(T_j\) represents the time at which the most recent \textit{informative} packet for process \(j\) was generated. The AoI for each process \(j\) evolves as follows: it increases linearly over time until an informative status update is successfully received, at which point a drop in the AoI occurs. However, whether an incoming packet contributes to reducing the AoI of a specific process depends on two key factors: (1) whether the packet contains information about process \(j\), and (2) the server's preemption dynamics.

To analyze this further, we model the server's operation through three states: \(0\) (idle), \(1\) (busy processing an informative packet), and \(2\) (busy processing an uninformative packet). In state \(0\), the server is not processing any packets, and the AoI for process \(j\) increases linearly due to the absence of updates. Upon transitioning to state \(1\), the server processes a packet containing relevant information for process \(j\), resulting in a decrease in the AoI after the packet's service time. Conversely, in state \(2\), the server is busy processing a packet that lacks relevant information, so the AoI for process \(j\) continues to increase linearly. Furthermore, when a new packet arrives at the busy server, incoming packets can interrupt ongoing service with a probability determined by the preemption matrix \(\bfp\). A transition from state \(2\) to state \(1\) via preemption leads to a linear increase in AoI during service time, followed by a decrease in AoI if the informative packet is successfully served, while a transition from state \(1\) to state \(2\) leads to a linear increase in AoI. The interaction between service states and preemption dynamics thus determines the AoI behavior over time. Lemma \ref{Lem1} provides the stationary distribution of these states and forms the foundation for deriving the closed-form expression of the AoI. %\ali{it feels like so much words, can we make this shorter? Also, why the results are given as a remark instead of a Lemma?}

%As per our system model, a served packet may or may not have information about process \(j\). If the served packet contains information on process \(j\), the AoI for process \(j\) decreases just after the service time of that packet. Conversely, if the served packet lacks information about process \(j\), the AoI for process \(j\) continues to increase linearly.

%Furthermore, the preemption mechanism impacts the AoI dynamics. Wrhen a new packet arrives at the server, it can preempt the currently served packet with a probability determined by the preemption matrix \(\bfp\). If preemption occurs, the ongoing service is interrupted, and the new packet is served instead. This can potentially reduce the AoI more effectively if the preempting packet contains relevant information for process \(j\). However, if the preempting packet lacks such information, the AoI will remain unaffected and continue its linear growth.

\begin{Lemma}\label{Lem1}
The stationary distribution of the states (\(0\), \(1\), \(2\)) can be derived as follows:
\begin{align}
\pi_0 = \frac{\mu}{(\lambda_C + \mu)}, \quad
\pi_1 = \frac{\lambda_C\Tilde{\lambda}_{1} + \dot{\lambda}_{1}\mu + \Tilde{\lambda}_{1}\mu}{(\lambda_C + \mu)(\Tilde{\lambda}_{C} + \mu)}, \\
\pi_2 = \frac{\Tilde{\lambda}_{C}\lambda_C + \lambda_C\mu -\lambda_C\Tilde{\lambda}_{1}  - \dot{\lambda}_{1}\mu - \Tilde{\lambda}_{1}\mu}{(\lambda_C + \mu)(\Tilde{\lambda}_{C} + \mu)}. 
\label{pi-distributions}
\end{align}
\begin{proof}
The illustration of the Markov chain states and the details of the proof can be found in \ifthenelse{\boolean{withappendix}}
{Appendix~\ref{spv-appendix}}
{Appendix B in \cite{technicalNote}}. %\ali{Perhaps you can mention here that an illustration of the Markov chain can also be found in the appendix}
\end{proof}
\end{Lemma}
With this model in place, we derive a closed-form expression for the AoI of each process, accounting for both informative and uninformative status updates with probabilistic preemption to comprehensively evaluate their impact on the system's dynamics. %as explained in Theorem \ref{The1}.

\begin{Theorem}\label{The1}
In the considered M/M/1/1 system, the average AoI for process $j$, denoted as $\Delta_j$, is

\footnotesize
\begin{align}
\Delta_j= \frac{\mu(\mu+\lambda_C)^2 + \sum_{i=1}^{N} \left(\mu\lambda_C \nc_{ij}(1-\np_{i}) + (\mu + \lambda_C)^2\np_{i}\right)\lambda_{i}}{
        \mu \sum_{i=1}^{N} \left(\mu + \lambda_C)(\lambda_C \np_{i} + \mu \right) \nc_{ij}\lambda_{i}}.
    \end{align}
\normalsize
    \end{Theorem}
\begin{proof} The proof leverages stochastic hybrid system modeling, focusing on state transitions (idle, busy with informative, or uninformative packets). Full details are provided in \ifthenelse{\boolean{withappendix}}
{Appendix~\ref{aoi-appendix}}
{Appendix C in \cite{technicalNote}}. 
\end{proof}
%\ali{both remark1 and theorem 1 have the same appendix?} \egemen{yes.}
Leveraging the above results, we examine two specific scenarios of interest:
\begin{itemize}

    \item \textbf{Preempt every packet} (\( \bfp = 1 \)):
    \begin{align}
    \Delta_j = \frac{{\lambda}_C+\mu}{\mu \tilde{\lambda}_j} .
    \end{align}
    
    \item \textbf{Preempt nothing} (\( \bfp = 0 \)):
    \begin{align}
    \Delta_j = \frac{{\lambda}_C}{\mu ({\lambda}_C + \mu)} + \frac{{\lambda}_C + \mu}{\mu \dot{\lambda}_j}.
    \end{align}

\end{itemize}

For these specific cases, both $\dot{\lambda}_j$ and $\tilde{\lambda}_j$ are equal to $\sum_{i=1}^{N} \nc_{ij}\lambda_i$. Therefore, the AoI in the no-preemption scenario is equal to the sum of a positive constant and the AoI in the full-preemption scenario. Thus, regardless of the correlation, preempting every packet guarantees a lower AoI than the no-preemption strategy. As the service rate approaches infinity, the average AoI difference between the two scenarios decreases because the system can accommodate updates almost instantaneously, which minimizes the necessity for preemption. Beyond these two special cases, we investigate the AoI-optimal preemption policy for our system in the following section.

\section{Average Age Optimization}\label{aoi-opt}

Preemption is a possible strategy for minimizing AoI in networked systems, but its effectiveness depends on the specific context, such as sensor arrival rates and the correlation between updates. In multi-sensor systems with correlated processes, the goal is to identify when and how preemption can provide benefits. Thus, our objective is to find optimal preemption probabilities that minimize AoI for a given setup.

%\ali{Let us reduce this; we speak too much. Just say, the goal is to find out in what scenario is preemption going to help in this multi-sensors correlated systems, etc.}
%\suresh{I actually like this rationale but would suggest making it more concise. And it sounds like this is motivation and belongs in the Intro.}

To that end, the objective is to minimize the sum of AoI, $\Delta_{\text{sum}}$, as follows:
\begin{align}
\min_{\bfp \in [0,1]^{N}} \sum_{j=1}^{M}\Delta_j = \Delta_{\text{sum}}.
\end{align}

%\suresh{Pay attention to notation, you have an $NxM$ matrix for $R$ here. Also, did you define $\Delta_{\text{sum}}$}

%\ali{align it better}
The objective function can be reformulated as:
\begin{align}\label{objective_func}
    \min_{\bfp \in [0,1]^{N}}\Delta_{\text{sum}} =& 
    \sum_{j=1}^{M} \frac{\mathbf{g_j}^T \bfp + g_{j0}}{\mathbf{f_j}^T \bfp + f_{j0}} = \sum_{j=1}^{M} \frac{G_j(\bfp)}{F_j(\bfp)}
    \\ \nonumber
    &\text{subject to } 0 < \bfp < 1,
\end{align}
where
\vspace{-15pt}
\begin{align}
    g_{j0} &= \mu(\mu + \lambda_C)^2 + \sum_{i=1}^{N} \lambda_{i}\mu\lambda_C \nc_{ij}, \\
    g_{ji} &= ((\mu + \lambda_C)^2 - \mu\lambda_C \nc_{ij})\lambda_{i}, \\
    f_{j0} &= (\mu + \lambda_C) \mu^2 \sum_{i=1}^{N} \nc_{ij} \lambda_{i}, \\
    f_{ji} &= \lambda_C \mu (\mu + \lambda_C) \nc_{ij} \lambda_{i}.
\end{align} 

This is a classical sum of linear ratios problem studied in the literature. The problem is well known for its computational challenges \cite{schaible2003fractional}, and it is NP-hard, in general \cite{freund2001solving}. However, given the special case of our problem, as we will show afterward, we can achieve a global optimum efficiently using a branch-and-bound algorithm with a finite number of iterations \cite{JIAO2022112701}.

%\ali{say that it is NP-hard in general. also, don't say progress has been made, etc. Just say, however, given the special case of our problem, as we will show afterwards, we can achieve a global optimum in an efficient way.}

To address the global minimization of the objective function in (\ref{objective_func}), an essential step is the reformulation of the original problem into an equivalent problem (EP). This reformulation allows for more tractable computation while preserving the global properties of the original problem.

For clarity, consider the following notations: for each \( i = 1, 2, \dots, M \), define
\[
\begin{aligned}
    \bar{l}_i &= \min_{\bfp \in [0,1]^{N}} F_i(\bfp), \quad 
    \bar{u}_i = \max_{\bfp \in [0,1]^{N}} F_i(\bfp), \\
    \bar{L}_i &= \min_{\bfp \in [0,1]^{N}} G_i(\bfp), \quad 
    \bar{U}_i = \max_{\bfp \in [0,1]^{N}} G_i(\bfp).
\end{aligned}
\]

% \ali{can you explain what linear problems are there?} 
These bounds can be determined by solving \( 4 \times M \) linear programs, where each program involves optimizing \( F_i(\bfp) \) and \( G_i(\bfp) \) over the feasible region \( \bfp \in [0,1]^{N} \), separately minimizing and maximizing each function to yield the values of \( \bar{l}_i \), \( \bar{u}_i \), \( \bar{L}_i \), and \( \bar{U}_i \). Using these results, we define the feasible region $
\Omega = \{ (\beta, \alpha) \in \mathbf{R}^{2M} \mid \bar{l}_i \leq \beta_i \leq \bar{u}_i, \; \bar{L}_i \leq \alpha_i \leq \bar{U}_i, \; i = 1, 2, \dots, M \}.
$
%\suresh{Wouldn't it be more natural to call it $(\alpha, \beta)$?} \egemen{I followed the notation in the cited paper. Do you want me to change?}\suresh{No; didn't know that was how it was done in that paper.}
The problem can then be reformulated as the following EP:

\[
\text{EP}:
\begin{cases}
    \min \; h(\beta, \alpha) = \sum_{i=1}^M \frac{\beta_i}{\alpha_i}, \\
    \text{s.t. } \quad F_i(\bfp) - \beta_i = 0, \; i = 1, 2, \dots, M, \\
    \quad G_i(\bfp) - \alpha_i = 0, \; i = 1, 2, \dots, M, \\
    \quad \bfp \in [0,1]^N, \; (\beta, \alpha) \in \Omega.
\end{cases}
\]

The feasible region of the EP, denoted as $
Z = \{ F_i(\bfp) - \beta_i = 0, \; G_i(\bfp) - \alpha_i = 0, \; i = 1, 2, \dots, M, \; \bfp \in [0,1]^N, \; (\beta, \alpha) \in \Omega \},
$
is a bounded compact set. Importantly, \( Z \neq \emptyset \) holds if and only if \( [0,1]^N \neq \emptyset \).
If a solution is globally optimal for the EP, it can be converted into a globally optimal solution for the original problem, and vice versa, making the EP sufficient for addressing the original optimization problem \cite{JIAO2022112701}.

Moreover, reformulating the problem as the EP allows us to analyze its computational complexity. Specifically, leveraging the EP's structure, we derive an upper bound on the iterations required for a branch-and-bound algorithm to find a global solution. The upper bound on the number of iterations required for a branch-and-bound algorithm to solve the EP is established in Theorem~\ref{Theo2}, as follows.

%\ali{This paragraph and the one afterward, it feels we are speaking too much; let us make things more compact with the major stuff that makes us appear that we have done quite some stuff here.} 

\begin{Theorem}\label{Theo2}
For any given positive error \( \epsilon_0 \in (0, 1) \), the outer space accelerating branch-and-bound algorithm can find a global \( \epsilon_0 \)-optimal solution in at most 
\begin{equation}
M \cdot \left\lceil \log_2 \frac{4M (\mu+\lambda_C)^2\lambda_C^2}{\epsilon_0\mu^3\hat{\lambda}_{\min}^{2}} \right\rceil
\end{equation}
iterations, where \( \hat{\lambda}_{\min} = \min(\boldsymbol{\lambda}^T \bfc) \) denotes the minimum element of \( \boldsymbol{\lambda}^T \bfc \) over the feasible set.
\end{Theorem}

\begin{proof} 
The proof uses the compact and bounded feasible region properties of the reformulated optimization problem to derive the computational complexity. Full details are provided in \ifthenelse{\boolean{withappendix}}
{Appendix~\ref{iteration-appendix}}
{Appendix D in \cite{technicalNote}}.
\end{proof}

This upper bound highlights the effect of parameters on the iteration count. Increasing $\mu$ reduces the number of required iterations because it also decreases the AoI value. Therefore, there is a need for a reduction in $\epsilon_0$ to achieve a similar level of precision. Likewise, smaller $\hat{\lambda}_{\min}$ and larger $\lambda_C$ lead to higher AoI values, necessitating more iterations to maintain the same error tolerance $\epsilon_0$ as expected. For an example case with $M = 10$, $\mu = 5$, $\lambda_C = 15$, and $\epsilon_0 = 0.01$, the upper bound requires at most 82 iterations, which is not excessively high. %\ali{You did not give a small example to showcase how fast the algorithm converge for typical values. Let us make our case stronger here}

\section{Numerical Results}\label{numerical}

In this section, we present numerical results to validate the theoretical analysis and optimization model developed in the previous sections. For clarity and better visualization, we consider a system with two sensors and two processes in our numerical experiments. However, it is important to note that all analyses can be generalized to systems with more sensors and processes. We vary different system parameters to verify our theoretical results in Section \ref{aoi-S}. After verifying the theoretical analysis, we investigate the optimization model described in Section \ref{aoi-opt} under different conditions and provide the results.

%\ali{I don't think we should talk about 2D grid search...people will say why did you do the analysis if u want to do grid search. Take this example back to after Theorem 2. And just say that we optimize the system and provide the results, and that's it. }, using a 2D grid search with a step size of 0.002 for both preemption probabilities, $\np_1$ and $\np_2$. While Theorem \ref{Theo2} provides an upper bound requiring at most 82 iterations for an example case with $M = 10$, $\mu = 5$, $\lambda_C = 15$, and $\epsilon_0 = 0.01$ that is not excessively high. However, the grid search is particularly suitable here due to the low system complexity.

Figure \ref{fig:exp_res} compares analytical results with experimental results for varying service rates, arrival rates, and preemption probabilities. Our simulations are unit-time-based and were run for 1 million units of time. The lowest arrival rate is 0.5 arrivals per unit time, ensuring at least $5 \times 10^5$ arrivals for each process to guarantee convergence. Simulation results are represented as circles, while the theoretical results derived from our analysis are depicted as solid lines in the figures. The strong alignment between the two verifies the validity of our analysis.

\begin{figure}[t!]
    \centering
    % Top-left subfigure
    \begin{subfigure}[t]{0.22\textwidth}
        \centering
        \includegraphics[width=1\textwidth]{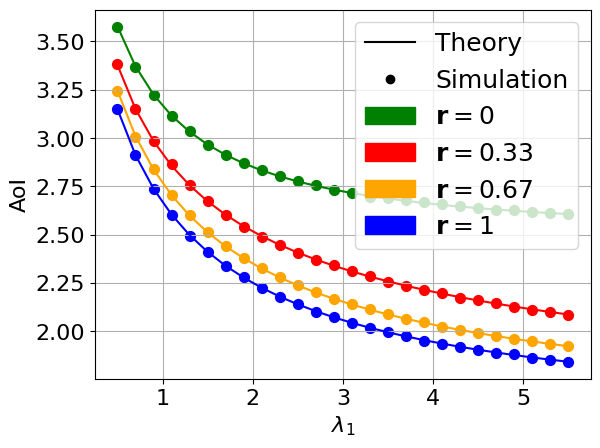}
        \caption{AoI versus $\lambda_1$ for different $\np_1=\np_2$ values with $\mu=2, \lambda_2=1,$ $\bfc= \begin{bmatrix}
1 & 0.5 \\
0.5 & 1
\end{bmatrix}$.}
        \label{fig:exp1}
    \end{subfigure}
    \hfill
    % Top-right subfigure
    \begin{subfigure}[t]{0.22\textwidth}
        \centering
        \includegraphics[width=1\textwidth]{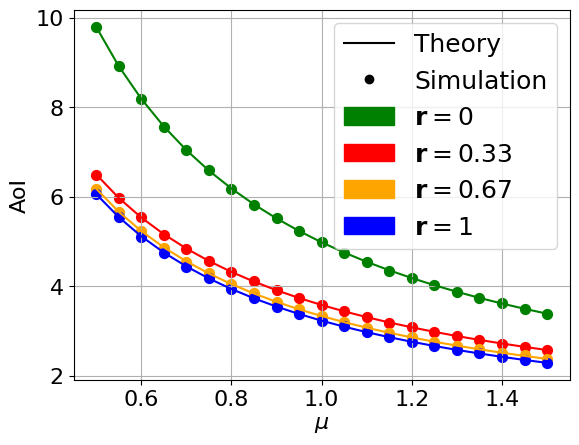}
        \caption{AoI versus $\mu$ for different $\np_1=\np_2$ values with $\lambda_1=1, \lambda_2=6,$ $\bfc= \begin{bmatrix}
1 & 0.5 \\
0.5 & 1
\end{bmatrix}.$}
        \label{fig:exp2}
    \end{subfigure}
    \caption{Simulation results vs theoretical findings.}
    \vspace{-14pt}
    \label{fig:exp_res}
\end{figure}

% \ali{Fix the figures after we decide on preemption probabilities; }

%\suresh{The captions for Fig. 3 (a) and (b) seem to be reversed.}

After that, we evaluate the optimal probabilistic preemption strategy for various correlations, arrival rates, and service rates. Figures \ref{l_res} and \ref{mu_res} illustrate the optimal probability values as well as their corresponding AoI values.  Additionally, the figures compare these optimal AoI results to the AoI values achieved under no preemption and full preemption scenarios. 

%\suresh{Clarify what you mean by aggregate updates.}

In Figure (\ref{fig:l1}), when the correlation matrix is an identity matrix, the preemption probability of Sensor 2 increases with $\lambda_1$ until reaching full preemption. On the other hand, Sensor 1’s probability starts high but decreases as $\lambda_1$ increases. This reflects a shift in importance toward Sensor 2 when Sensor 1’s arrival rate is higher. This behavior illustrates that packets with low arrival rates should preempt those with higher arrival rates when the correlation effect is eliminated. On the other hand, Figure (\ref{fig:l2}) shows the effect of the correlation matrix. Sensor 1 fully dominates due to its complete information about both processes and even at high \(\lambda_1\), preempting other packets with a packet from Sensor 2 is not preferable since it lacks any information about Process 1. In addition, Figure (\ref{fig:theta}) illustrates how optimal preemption evolves when each sensor fully tracks one process and partially observes the other with probability $\theta$. When $\theta$ is small, the sensor with the lower arrival rate is given higher preemption priority. However, as $\theta$ increases, the sensors become more similar in the information they provide. This reduces the distinction between them, leading to an optimal strategy of preempting every packet. These results underscore how the correlation matrix governs the preemption strategy and sensor roles.

\begin{figure}[h!]
    \centering
    % Top-left subfigure
    \begin{subfigure}[b]{0.24\textwidth}
        \centering
        \includegraphics[width=1\textwidth]{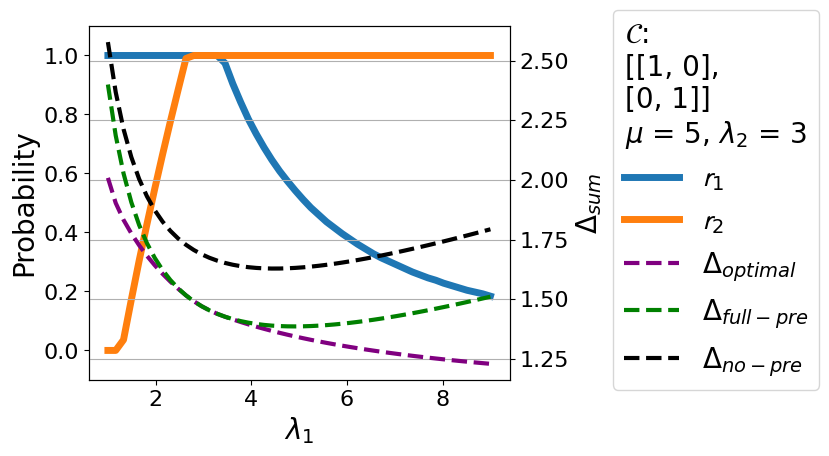}
        \caption{$\frac{}{}$ $\frac{}{}$$\frac{}{}$ $\frac{}{}$$\frac{}{}$$\frac{}{}$ $\frac{}{}$$\frac{}{}$ $\frac{}{}$}
        \label{fig:l1}
    \end{subfigure}
    \hfill
    % Top-right subfigure
    \begin{subfigure}[b]{0.24\textwidth}
        \centering
        \includegraphics[width=1\textwidth]{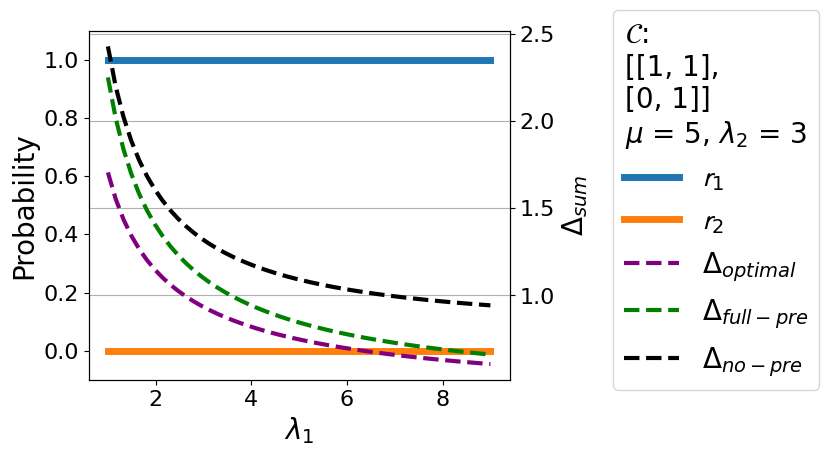}
        \caption{$\frac{}{}$ $\frac{}{}$$\frac{}{}$ $\frac{}{}$$\frac{}{}$$\frac{}{}$ $\frac{}{}$$\frac{}{}$ $\frac{}{}$}
        \label{fig:l2}
    \end{subfigure}
% \vspace{-10pt}
%     \vskip\baselineskip
%     % Bottom-left subfigure
%     \begin{subfigure}[b]{0.22\textwidth}
%         \centering
%         \includegraphics[width=0.9\textwidth]{figures/9.png}
%         \caption{}
%         \label{fig:l3}
%     \end{subfigure}
%     \hfill
%     % Bottom-right subfigure
%     \begin{subfigure}[b]{0.22\textwidth}
%         \centering
%         \includegraphics[width=0.9\textwidth]{figures/10.png}
%         \caption{}
%         \label{fig:l4}
%     \end{subfigure}
    \caption{Optimal preemption probabilities under varying $\lambda_1$ to show the effect of correlation matrix.}
    \label{l_res}
    \vspace{-8pt}
\end{figure}

Lastly, in Figure (\ref{fig:div}), we see that Sensor 2 has more priority (higher preemption probability) than Sensor 1, even though Sensor 1 has more aggregate updates on average, i.e., 
the sum of the elements in the corresponding row of the correlation matrix is larger for Sensor 1. This stems from Sensor 2 having more diverse updates and playing a bigger role in minimizing the sum AoI. Note that in a single-source memoryless system, preempting packets is generally beneficial for reducing the AoI. Specifically, self-preempting, where a packet preempts another packet with information from the same processes, is profitable in our system. An increase in preemption probability results in a gain from self-preemption. However, there is a trade-off because it also increases the probability of preempting an informative packet for additional or different processes that can negatively impact the AoI. When $\lambda_1$ is low, the benefit of self-preemption for Sensor 1 is higher than the loss of preempting a packet from Sensor 2. However, as $\lambda_1$ increases, the risk of losing informative packets from other processes also rises when a packet from Sensor 1 preempts. Consequently, the optimal preemption probability decreases as $\lambda_1$ grows. Notably, this decrease begins even when $\lambda_1$ is smaller than $\lambda_2$, making packets from Sensor 1 relatively rare. Nevertheless, the diversity of information provided by Sensor 2 makes its packets more valuable, even when the aggregate correlation is low.
\vspace{-12pt}

%it is better to let only packets from Sensor 2 preempt because there might be a risk of losing a packet with more information when a packet from Sensor 1 preempts. With the increase in $\mu$, service time reduces, and packets from Sensor 1 start preempting with some probability to obtain the optimal AoI because it reduces to AoI by preempting their own packets. However, this probability decreases after some threshold and goes to zero again. The reason is that when $\mu$ is high, the gain of self-preemption is lower than the loss of preempting a packet with more information.

\begin{figure}[h!]

    \centering
    % Top-left subfigure
     \begin{subfigure}[b]{0.24\textwidth}
        \centering
        \includegraphics[width=1\textwidth]{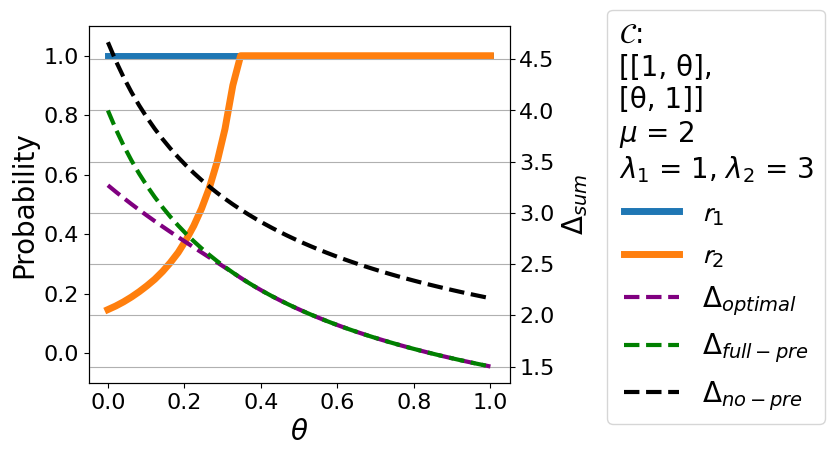}
        \caption{$\frac{}{}$$\frac{}{}$$\frac{}{}$$\frac{}{}$$\frac{}{}$$\frac{}{}$ $\frac{}{}$$\frac{}{}$ $\frac{}{}$}
        \label{fig:theta}
    \end{subfigure}
    \hfill
    % Top-right subfigure
    \begin{subfigure}[b]{0.24\textwidth}
        \centering
        \includegraphics[width=1\textwidth]{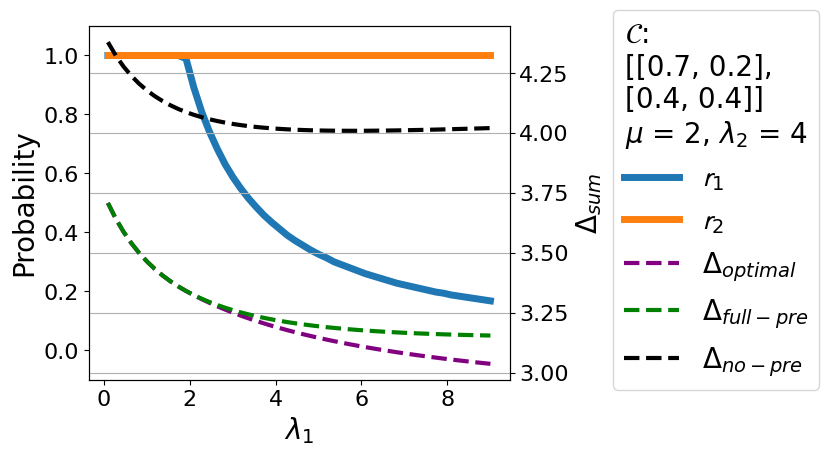}
        \caption{$\frac{}{}$$\frac{}{}$ $\frac{}{}$$\frac{}{}$$\frac{}{}$$\frac{}{}$ $\frac{}{}$$\frac{}{}$ $\frac{}{}$}
        \label{fig:div}
    \end{subfigure}   
\vspace{-15pt}
    \caption{Optimal preemption probabilities under different conditions.}
\vspace{-10pt}
    \label{mu_res}
\end{figure}

%\suresh{Make the lines bolder.}

%\suresh{I'm not sure this figure is interesting anymore; the diff between full-pre and optimal-pre AoI is tiny.}

%\suresh{There is so much theory on managing the complexity of the optimization. Would be good to present some results of that here. It's not even clear what the $\epsilon_0$ you used to get the results and how many iterations were needed.} \egemen{You are right. However, we use grid search to find the optimum here. How should I mention it?}\suresh{I'm not sure. How many iterations does gridsearch take and how does it relate to the UB you derived. Seems that your optimization is disconnected from the result of Thm 2. Is that the case?} \egemen{Yes, the UB is for a branch and bound algorithm. However, for this experiment, we check every possible point with step size 0.002 to find the optimal probability.}

\vspace{-4pt}
\section{Conclusion}\label{conc}

In this paper, we proposed and analyzed a probabilistic preemption strategy to minimize the AoI in correlated multi-sensor multi-process systems. We derived closed-form expressions for the AoI by applying stochastic hybrid system modeling. Next, we optimized source-specific preemption probabilities to minimize the sum AoI. The problem was framed as the sum of linear ratios, which is generally NP-hard. To address this, we reformulated the problem and provided an upper bound on the number of iterations required to find the optimal solution. This upper bound shows that the optimal scenario requires relatively few iterations. We verified our theoretical findings through experiments and demonstrated the optimal preemption strategy under various conditions. The results show that optimal preemption strategies can differ significantly based on the correlation matrix. Interestingly, packets with more diverse information probability are more important than packets with more information probability.

\section*{Acknowledgment}
This work was supported in part by NSF grant CNS-2219180.
\newpage

\bibliographystyle{ieeetr}
\bibliography{references}
\ifthenelse{\boolean{withappendix}}
{\appendices
\section{} 
\label{reduction-P}

To prove our argument, we apply the splitting property of the Poisson process. Let \( N(t) \) be a Poisson process with rate parameter \( \lambda \). If events are split into two groups with probabilities \( p \) and \( 1-p \), then the resulting processes \( N_1(t) \) and \( N_2(t) \) are independent Poisson processes with rate parameters \( p\lambda \) and \( (1-p)\lambda \) respectively \cite{splitting_poisson}.

From process \( j \)'s perspective, we can split arrivals from sensor \( i \) into two groups: informative and uninformative arrivals with probabilities \( \nc_{ij} \) and \( 1-\nc_{ij} \), respectively. The rate of arrivals from sensor \( i \) is \( \lambda_i \), so the rate of informative arrivals for process \( j \) from sensor \( i \) is \( \nc_{ij}\lambda_i \). Additionally, we can further split the informative arrivals based on whether they can preempt ongoing service. The rate of informative arrivals that can preempt ongoing service for process \( j \) from sensor \( i \) is \( \np_{i}\nc_{ij}\lambda_i \) and the rate of informative arrivals that can not preempt ongoing service for process \( j \) from sensor \( i \) is \( (1-\np_{i})\nc_{ij}\lambda_i \). Since all these arrivals are Poisson, we can merge them into a single process. The total arrival rate of informative packets that can preempt ongoing service for process \( j \) is given by

\begin{equation}
\Tilde{\lambda}_j = \sum_{i=1}^{N} \np_{i}\nc_{ij}\lambda_i
\end{equation}

Similarly, the total arrival rate of informative packets that can not preempt ongoing service for process \( j \) is

\begin{equation}
\Tilde{\lambda}_j = \sum_{i=1}^{N} (1-\np_{i})\nc_{ij}\lambda_i
\end{equation}

We can express these rates in vector form as follows:

\begin{equation}
\boldsymbol{\Tilde{\lambda}}^T = \begin{bmatrix}
\Tilde{\lambda}_{1} & \Tilde{\lambda}_{2} & \dots & \Tilde{\lambda}_{M}
\end{bmatrix} = (\boldsymbol{\lambda}^T \odot \bfp^T) \bfc,
\end{equation}
\begin{equation}
\boldsymbol{\dot{\lambda}}^T = \begin{bmatrix}
\dot{\lambda}_{1} & \dot{\lambda}_{2} & \dots & \dot{\lambda}_{M}
\end{bmatrix} = (\boldsymbol{\lambda}^T \odot (1-\bfp^T)) \bfc,
\end{equation}

The importance of the packet is whether it has information of process $j$ so  we can say that The system with $N$ sensors and arrival rates $\boldsymbol{\lambda}$ shown in Figure \ref{fig:system_model} equivalents to the system with two sources as shown in Figure \ref{fig:equiv_model} from process $j$'s perspective.

\section{}\label{spv-appendix}

We adopt the stochastic hybrid system (SHS) model as defined in \cite{yates2019}, with a key distinction: our model incorporates probabilistic preemption. The system dynamics are depicted in Figure \ref{fig:equiv_model} so we can analyze the AoI for any process $i$ and generalize it. First, the discrete state is denoted as $q(t) = q \in Q = \{0, 1, 2\}$, where $q = 0$ represents an idle server, and $q \in \{1, 2\}$ signifies that an update packet is currently being serviced. The continuous state is described as $x(t) = [x_0(t), x_1(t)]$, where $x_0(t)$ represents the current age of the process, and $x_1(t)$ captures the potential age if the packet in service is successfully delivered. Notably, $x_1(t)$ is irrelevant in state $0$ since no packet is in service. In state $1$, $x_1(t)$ corresponds to the age of the informative update being serviced. Conversely, in state $2$, where an uninformative update is in service, the completion of this update does not affect the process age, rendering $x_1(t)$ irrelevant in this state as well.

\begin{table}[h]
\centering
\caption{Table of Transitions for the Markov Chain in Figure \ref{fig:shs}.}
\begin{tabular}{c c c c c c}
\toprule
$l$ & $q_l \rightarrow q'_l$ & $\lambda^{(l)}$ & $\mathbf{xA}_l$ & $\mathbf{A}_l$ & $\mathbf{v}_{q_l}\mathbf{A}_l$ \\
\midrule
1 & $0 \rightarrow 1$ & $\Tilde{\lambda}_{1}+\dot{\lambda}_{1}$ & $\begin{bmatrix} x_0 & 0 \end{bmatrix}$ & \small $\begin{bmatrix} 1 & 0 \\ 0 & 0 \end{bmatrix}$ \normalsize & $\begin{bmatrix} v_{00} & 0 \end{bmatrix}$ \\
2 & $0 \rightarrow 2$ & $\lambda_{C}-\Tilde{\lambda}_{1}-\dot{\lambda}_{1}$ & $\begin{bmatrix} x_0 & 0 \end{bmatrix}$ & \small $\begin{bmatrix} 1 & 0 \\ 0 & 0 \end{bmatrix}$ \normalsize & $\begin{bmatrix} v_{00} & 0 \end{bmatrix}$ \\
3 & $1 \rightarrow 0$ & $\mu$        & $\begin{bmatrix} x_1 & 0 \end{bmatrix}$ & \small$\begin{bmatrix} 0 & 0 \\ 1 & 0 \end{bmatrix}$ \normalsize & $\begin{bmatrix} v_{11} & 0 \end{bmatrix}$ \\
4 & $1 \rightarrow 1$ & $\Tilde{\lambda}_{1}
$  & $\begin{bmatrix} x_0 & 0 \end{bmatrix}$ & \small$\begin{bmatrix} 1 & 0 \\ 0 & 0 \end{bmatrix}$\normalsize & $\begin{bmatrix} v_{10} & 0 \end{bmatrix}$ \\
5 & $1 \rightarrow 2$ & $\Tilde{\lambda}_{C}-\Tilde{\lambda}_{1}$  & $\begin{bmatrix} x_0 & 0 \end{bmatrix}$ & \small$\begin{bmatrix} 1 & 0 \\ 0 & 0 \end{bmatrix}$ \normalsize & $\begin{bmatrix} v_{10} & 0 \end{bmatrix}$ \\
6 & $2 \rightarrow 0$ & $\mu$        & $\begin{bmatrix} x_0 & 0 \end{bmatrix}$ & \small$\begin{bmatrix} 0 & 0 \\ 1 & 0 \end{bmatrix}$ \normalsize & $\begin{bmatrix} v_{20} & 0 \end{bmatrix}$ \\
7 & $2 \rightarrow 1$ & $\Tilde{\lambda}_{1}$  & $\begin{bmatrix} x_0 & 0 \end{bmatrix}$ & \small$\begin{bmatrix} 1 & 0 \\ 0 & 0 \end{bmatrix}$ \normalsize & $\begin{bmatrix} v_{20} & 0 \end{bmatrix}$ \\
\bottomrule
\end{tabular}
\label{shs_table}
\end{table}

\begin{figure}
    \centering
    \includegraphics[width=0.5\linewidth]{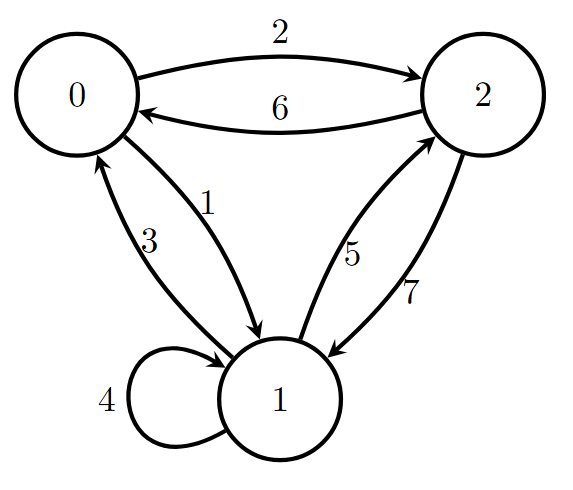}
    \caption{The Markov chain for updates.}
    \label{fig:shs}
\end{figure}

A Markov chain representing the discrete state $q(t)$ is depicted in Figure~\ref{fig:shs}. The corresponding transitions of the SHS at state $q_l$ are detailed in Table~\ref{shs_table}. In the figure, a directed edge $l$ from node $q$ to node $q'$ indicates that transitions from state $q$ to state $q'$ occur at an exponential rate $\lambda^{(l)}$, as specified in the table.

%\suresh{Incomplete.}

We first show that the stationary probability vector $\pi$ satisfies $
\mathbf{\pi D} = \mathbf{\pi Q} \quad \text{with}$ 
\begin{align}
\quad
\mathbf{D} = \text{diag}[\lambda_{C}, \mu + \Tilde{\lambda}_{C}, \mu + \Tilde{\lambda}_{1}], \quad  \\ \mathbf{Q} = 
\begin{bmatrix}
0 & \Tilde{\lambda}_{1}+\dot{\lambda}_{1} & \lambda_{C}-\Tilde{\lambda}_{1}-\dot{\lambda}_{1} \\
\mu & \Tilde{\lambda}_{1} & \Tilde{\lambda}_{C}-\Tilde{\lambda}_{1} \\
\mu & \Tilde{\lambda}_{1} & 0
\end{bmatrix}.
\end{align}
Applying $\sum_{i=0}^{2} \pi_i = 1$, the stationary probabilities are 
\begin{equation}
\pi_0 = \frac{\mu}{(\lambda_C + \mu)}, \label{pi0}
\end{equation}
\begin{equation}
\pi_1 = \frac{\lambda_C\Tilde{\lambda}_{1} + \dot{\lambda}_{1}\mu + \Tilde{\lambda}_{1}\mu}{(\lambda_C + \mu)(\Tilde{\lambda}_{C} + \mu)}, \label{pi1}
\end{equation}
\begin{equation}
\pi_2 = \frac{\Tilde{\lambda}_{C}\lambda_C + \lambda_C\mu -\lambda_C\Tilde{\lambda}_{1}  - \dot{\lambda}_{1}\mu - \Tilde{\lambda}_{1}\mu}{(\lambda_C + \mu)(\Tilde{\lambda}_{C} + \mu)} . \label{pi2}
\end{equation}

\section{}\label{aoi-appendix}

Given the SHS model and $\pi$ in Appendix \ref{spv-appendix}, we can evaluate $\bar{v}$ to find the AoI. Let 
\begin{equation}
\mathbf{\bar{v}} = [\mathbf{\bar{v}_0} \ \mathbf{\bar{v}_1} \ \mathbf{\bar{v}_2}] = [\bar{v}_{00} \ \bar{v}_{01} \ \bar{v}_{10} \ \bar{v}_{11} \ \bar{v}_{20} \ \bar{v}_{21}].   
\end{equation}
It follows that
\begin{equation}
\mathbf{\bar{v}D} = \mathbf{\pi B} +  \mathbf{\bar{v}R},
\end{equation}
where 
\begin{equation}
\mathbf{D} = \text{diag}[\lambda_C, \lambda_C, \mu + \Tilde{\lambda}_{C}, \mu + \Tilde{\lambda}_{C}, \mu + \Tilde{\lambda}_{1}, \mu + \Tilde{\lambda}_{1}],
\end{equation}
\begin{equation}
\mathbf{B} =
\begin{bmatrix}
1 & 0 & 0 & 0 & 0 & 0 \\
0 & 0 & 1 & 1 & 0 & 0 \\
0 & 0 & 0 & 0 & 1 & 0
\end{bmatrix},
\end{equation}
and
\begin{equation}
\mathbf{R} = 
\begin{bmatrix}
0 & 0 & \Tilde{\lambda}_{1}+\dot{\lambda}_{1}  & 0 & \lambda_{C}-\Tilde{\lambda}_{1}-\dot{\lambda}_{1} & 0 \\
0 & 0 & 0 & 0 & 0 & 0 \\
0 & 0 & \Tilde{\lambda}_{1} & 0 & \Tilde{\lambda}_{C}-\Tilde{\lambda}_{1} & 0 \\
\mu & 0 & 0 & 0 & 0 & 0 \\
\mu & 0 & \Tilde{\lambda}_{1} & 0 & 0 & 0 \\
0 & 0 & 0 & 0 & 0 & 0
\end{bmatrix}.
\end{equation}

Then, we obtain $\bar{v}_{01}=\bar{v}_{21} = 0 $ and 
\begin{align}
&
\label{pi_v}
\begin{bmatrix}
\bar{\pi}_0 & \bar{\pi}_1 & \bar{\pi}_1 & \bar{\pi}_2
\end{bmatrix}
= \\ \nonumber \hat{\mathbf{v}}&
\begin{bmatrix}
\lambda_{C} & -\Tilde{\lambda}_{1}-\dot{\lambda}_{1} & 0 & \Tilde{\lambda}_{1}+\dot{\lambda}_{1}-\lambda_{C} \\
0 & \mu + \Tilde{\lambda}_{C}-\Tilde{\lambda}_{1} & 0 & \Tilde{\lambda}_{1} - \Tilde{\lambda}_{C} \\
-\mu & 0 & \mu + \Tilde{\lambda}_{C} & 0 \\
-\mu & -\Tilde{\lambda}_{1} & 0 & \mu + \Tilde{\lambda}_{1}
\end{bmatrix}, \\
\text{where } 
\hat{\mathbf{v}} &= 
\begin{bmatrix}
\bar{v}_{00} & \bar{v}_{10} & \bar{v}_{11} & \bar{v}_{20}
\end{bmatrix}. \nonumber
\end{align}

After solving eq. (\ref{pi_v}) using eqs. (\ref{pi0}), (\ref{pi1}), and (\ref{pi2}), we determine $\mathbf{\bar{v}}$. Later, we find the average age of information using the formula for a single process $j$ $\Delta_j = \sum_{q=0}^2 \bar{v}_{10}$ as follows:

%\suresh{Is this what you defined as $\Delta_{\rm sum}$ earlier?}

\footnotesize
\begin{align}
\Delta_j = \frac{\lambda_{C}^{2} \tilde{\lambda}_C + \lambda_{C}^{2} \mu + \lambda_{C} \dot{\lambda}_1 \mu + 2 \lambda_{C} \tilde{\lambda}_C \mu + 2 \lambda_{C} \mu^{2} + \tilde{\lambda}_C \mu^{2} + \mu^{3}}{\mu \left(\lambda_{C}^{2} \tilde{\lambda}_1 + \lambda_{C} \dot{\lambda}_1 \mu + 2 \lambda_{C} \tilde{\lambda}_1 \mu + \dot{\lambda}_1 \mu^{2} + \tilde{\lambda}_1 \mu^{2}\right)}
\end{align}
\normalsize

\section{}\label{iteration-appendix}

In this section, we discuss the upper bound on the number of iterations required by the outer space accelerating branch-and-bound algorithm to achieve a global $\epsilon_0$-optimal solution. According to Theorem 5 in \cite{JIAO2022112701}, for any given positive error $\epsilon_0 \in (0, 1)$, the algorithm converges to the desired solution in at most
\begin{equation}
p \cdot \left\lceil \log_2 \frac{p\tau \delta(\Omega)}{\epsilon_0} \right\rceil 
\end{equation}
iterations.

Here, the symbols used in the theorem are defined as follows:

\begin{itemize}
    \item $\Omega \subseteq \mathbf{R}^p$ is a compact hyper-subrectangle, and $\delta(\Omega)$ is defined as:
    \begin{equation}
    \delta(\Omega) = \max_{i=1,2,\dots,p} \{ \bar{U}_i - \bar{L}_i \},    
    \end{equation}
    where $\bar{U}_i$ and $\bar{L}_i$ represent the upper and lower bounds of the $i$-th dimension of the rectangle $\Omega$.

    \item $\tau$ is defined as:
    \begin{equation}\label{tau_eq}
    \tau = \max_{i=1,\dots,p} \frac{4 \max\{|\bar{l}_i|, |\bar{u}_i|\}}{\min\{\bar{L}_i, \bar{U}_i, \bar{L}_i^2, \bar{U}_i^2\}},
    \end{equation}
    where the terms are determined as follows:
    \begin{align}
        \bar{l}_i &= \min_{y \in \Theta} n_i(y), \quad \bar{u}_i = \max_{y \in \Theta} n_i(y), \nonumber \\
        \bar{L}_i &= \min_{y \in \Theta} d_i(y), \quad \bar{U}_i = \max_{y \in \Theta} d_i(y).
    \end{align}

    \item The terms $n_i(y)$ and $d_i(y)$ come from the problem defined as:
    \begin{align}
    \quad \min f(y) = \sum_{i=1}^p \frac{n_i(y)}{d_i(y)}, \quad \nonumber \\ \text{s.t.} \; y \in \Theta = \{y \in \mathbf{R}^n \mid Ay \leq b \}. 
    \end{align}
    \end{itemize}

We can reformulate our problem to determine the upper bound using these definitions. The variable in our problem is $\bfp$, and the objective is specified in (\ref{objective_func}). There are $M$ different linear fractions in the objective. The numerators of these fractions increase as any element of $\bfp$ increases. Consequently, we obtain $\bar{l}_i$ when $\bfp = 0$ and $\bar{u}_i$ when $\bfp = 1$ as follows:
\begin{align}
            \bar{l}_i &= \mu(\mu + \lambda_C)^2 + \sum_{i=1}^{N} \lambda_{i}\mu\lambda_C \nc_{ij},\quad \bar{u}_i = (\mu + \lambda_{C})^3
\end{align}

 Similarly, the denominators of these fractions decrease as any element of $\bfp$ increases, leading to $\bar{L}_i$ when $\bfp = 0$ and $\bar{U}_i$ when $\bfp = 1$.
 \begin{align}
            \bar{L}_i &=  (\mu + \lambda_C) \mu^2 \sum_{i=1}^{N} \nc_{ij} \lambda_{i}, \quad \bar{U}_i =  (\mu + \lambda_C)^2 \mu \sum_{i=1}^{N} \nc_{ij} \lambda_{i}.
\end{align}

After that, $\delta(\Omega)$ becomes: 

\begin{align}
        \delta(\Omega) = \max_{i=1,2,\dots,M} \{(\mu + \lambda_C)\lambda_C \mu \sum_{i=1}^{N} \nc_{ij} \lambda_{i}\} \leq (\mu + \lambda_C)\lambda_C^2 \mu , 
\end{align}

Last, we find $\tau$. In our problem, all parameters and variables are positive, so both the nominators and the denominators are positive, which can help us simplify eq. (\ref{tau_eq}) and obtain $\tau$ as follows:

    \begin{align}
    \tau = \max_{i=1,\dots,M} \frac{4 \bar{u}_i}{\bar{L}_i^2} = \frac{4 (\mu + \lambda_{C})}{\mu^4 \hat{\lambda}_{\min}^2},\\ \nonumber
    \text{where } \hat{\lambda}_{\min} = \min(\boldsymbol{\lambda}^T \bfc)
    \end{align}

Putting all together, for any given positive error $\epsilon_0 \in (0, 1)$, the outer space accelerating branch-and-bound algorithm can seek out a global $\epsilon_0$-optimum solution in at most 
\begin{equation}
M \cdot \left\lceil \log_2 \frac{4M (\mu+\lambda_C)^2\lambda_C^2}{\epsilon_0\mu^3\hat{\lambda}_{\min}^{2}} \right\rceil
\end{equation}
iterations as shown in Theorem \ref{Theo2}.}
{}

\end{document}